\documentclass[cmp,onecolumn]{svjour}
\usepackage{amsmath}
\usepackage{amssymb}
\usepackage{bbm}
\usepackage{enumerate}

\newcommand{\ket}[1]{\ensuremath{\vert#1\rangle}}
\newcommand{\bra}[1]{\ensuremath{\langle #1|}}

\newcommand{\XZ}{\ensuremath{X\!Z}}

\usepackage[normalem]{ulem}
\newcommand{\ens}[0]{\ensuremath}
\newcommand{\SkPr}[2]{\ens{\left\langle#1|#2\right\rangle}} 
\newcommand{\symp}[2]{\ens{\langle#1,#2\rangle_{\mathrm{sp}}}} 
\newcommand{\x}[0]{\ens{\otimes}} 
\newcommand{\Fkt}[3]{\ens{#1 : #2 \rightarrow #3}} 
\newcommand{\iE}[0]{\ens{\mathrm{i}}}
\newcommand{\EZ}[0]{\ens{\mathrm{e}}} 
\newcommand{\Eins}[0]{\ens{\mathbbm{1}}}
\newcommand{\ny}[0]{\ens{\nu}} 
\newcommand{\F}[0]{\ens{\mathbb{F}}}
\newcommand{\N}[0]{\ens{\mathbb{N}}}

\newcommand{\C}[0]{\ens{\mathbb{C}}}

\newcommand{\cB}[0]{\ens{\mathcal{B}}}
\newcommand{\cC}[0]{\ens{\mathcal{C}}}
\newcommand{\cH}[0]{\ens{\mathcal{H}}}
\newcommand{\fC}[0]{\ens{\mathfrak{C}}}
\newcommand{\fS}[0]{\ens{\mathfrak{S}}}
\newcommand{\Mge}[2]{\ens{\left\lbrace #1|\,#2 \right\rbrace}}
\newcommand{\Mg}[1]{\ens{\left\lbrace #1 \right\rbrace}}
\newcommand{\MgN}[1]{\ens{\Mg{0,\dots,#1}}}
\newcommand{\MgE}[1]{\ens{\Mg{1,\dots,#1}}}
\newcommand{\betrag}[1]{\ens{\left|#1\right|}}
\newcommand{\Erz}[1]{\ens{\langle #1 \rangle}}

\smartqed

\begin{document}

\title{Cyclic mutually unbiased bases, Fibonacci polynomials and Wiedemann's conjecture}
\titlerunning{Cyclic MUBs, Fibonacci polynomials and Wiedemann's conjecture}

\author{Ulrich Seyfarth\inst{1}
  \and Kedar S. Ranade\inst{2}}
\institute{Institut f\"ur Angewandte Physik, Technische Universit\"at Darmstadt, Hochschulstra\ss{}e 4a,
  64289 Darmstadt, Germany, \email{ulrich.seyfarth@physik.tu-darmstadt.de} \and Institut f\"ur Quantenphysik, Universit\"at Ulm, Albert-Einstein-Allee 11, 89081 Ulm, Germany, \email{kedar.ranade@uni-ulm.de}}
\date{Dated: March 27th, 2012}

%

\def\makeheadbox{}
\maketitle

\begin{abstract}
  We relate the construction of a complete set of cyclic mutually unbiased bases, i.\,e., mutually unbiased bases
  generated by a single unitary operator, in power-of-two dimensions to the problem of finding
  a symmetric matrix over $\F_2$ with an irreducible characteristic polynomial that has a given Fibonacci index.
  For dimensions of the form $2^{2^k}$ we present a solution that shows an analogy to an open conjecture
  of Wiedemann in finite field theory. Finally, we discuss the equivalence of mutually unbiased bases.
\end{abstract}

\section{Introduction}\label{sec:prelim}
\emph{Mutually unbiased bases} (MUBs) are used in various contexts in quantum information theory. The construction
of MUBs is connected to the theory of finite fields \cite{WF89,KR04} and to construction problems of Latin squares
\cite{PDB09} or of maximally commuting bases of orthogonal unitary matrices \cite{BBRV02}. In this
article we aim to construct complete sets of \emph{cyclic} mutually unbiased bases \cite{Chau,Gow}, which exist
only for even prime-power dimensions \cite{Appleby}. We start with the definition of MUBs.
\begin{definition}[Mutually unbiased bases]\label{def:MUB}\hfill\\
  Let $\cH = \C^d$ be a $d$-dimensional complex Hilbert space. Two orthonormal bases
  $\cB_1 = \Mg{\ket{\psi_1},\,\dots,\,\ket{\psi_d}}$ and $\cB_2 = \Mg{\ket{\phi_1},\,\dots,\,\ket{\phi_d}}$
  of $\cH$ are said to be \emph{mutually unbiased}, if there holds $\betrag{\SkPr{\psi_i}{\phi_j}}^2 = d^{-1}$
  for all $i,\,j \in \MgE{d}$.
\end{definition}
The standard examples for MUBs are the $z$-, $x$- and $y$-basis for $d = 2$ or any basis and its Fourier
transform for any dimension $d$. It is well-known that the maximum number of pairwise
mutually unbiased bases of $\C^d$, denoted $N(d)$, is upper-bounded by $d+1$. If $d$ is a prime
power, then $N(d) = d+1$ bases can be constructed by using methods from finite field theory.
In any other dimension, the precise number $N(d)$ is unknown.
\par A basis of $\cH = \C^d$ can always be identified with a unitary matrix $U$ on this space, since the
columns of such matrix define an orthonormal basis and vice versa. Since $U^2$, $U^3$ etc. are unitary,
they also define bases of $\cH$.
\par In the following, let $\N$ be the natural numbers and $\N_0 := \N \cup \Mg{0}$.
For a finite field of order $q$, we write $\F_q$, for the complex field $\C$, and for the $n \times n$-matrices
over an arbitrary field (or ring) $K$, we write $M_n(K)$.
\begin{definition}[Cyclic mutually unbiased bases]\hfill\\
A set $\Mg{\cB_0,\,\dots,\,\cB_{n-1}}$ of mutually unbiased bases is called \emph{cyclic}, if there exists
a unitary matrix $U \in M_d(\C)$ of finite order $n \in \N$, such that the columns of the matrices
$U,\,U^2,\,U^3,\,\dots,U^n = U^0 = \Eins_d$ coincide with the bases $\Mg{\cB_0,\,\dots,\,\cB_{n-1}}$.
\end{definition}
We are interested in \emph{complete sets of cyclic MUBs}, i.\,e. sets of $N(d)$ cyclic MUBs, and want
to construct a unitary generator $U$. Since $N(d)$ is known for prime powers only and it was shown that such
operators cannot exist for odd~$d$~\cite{Appleby}, we shall only consider $d = 2^m$ for some $m \in \N$.
(We will keep this notion of $d$ and $m$ throughout this article.)

\section{Methods and formalism}\label{sec:constr}
In this section, we introduce the methods which are relevant for our treatment of cyclic MUBs. We will
start with a subsection on Fibonacci polynomials and then relate the construction of cyclic MUBs
to properties of a particular matrix.

\subsection{Fibonacci polynomials}
Our analysis of complete sets of cyclic MUBs is closely related to the properties of the so-called
\emph{Fibonacci polynomials}.\footnote{Note that these polynomials differ by one from the $f_n$
used previously \cite{KRS10}, i.\,e. $F_n = f_{n-1}$.}
\begin{definition}[Fibonacci polynomials]\hfill\\
  Over an arbitrary field $K$, we define the \emph{Fibonacci polynomials} $(F_n)_{n \in \N_0}$,
  by $F_0(x) := 0$, $F_1(x) := 1$ and the recursion $F_{n+1}(x) := x \cdot F_n(x) + F_{n-1}(x)$.
\end{definition}
If $K = \C$, the sequence $(F_n(1))_{n \in \N_0}$ is the usual Fibonacci sequence. In the context of
cyclic MUBs we exclusively deal with the ground field $K = \F_2$ and possibly its extensions and make
use of a block matrix $C = {\tiny \begin{pmatrix} B & \Eins \\ \Eins & 0 \end{pmatrix}} \in M_{2m}(\F_2)$
with $\Eins,\,0 \in M_m(\F_2)$ and some $B \in M_m(\F_2)$. The powers of $C$ are easily seen to be of the form
\begin{equation}\label{C-Fib}
  C^k = \begin{pmatrix} F_{k+1}(B) & F_k(B) \\ F_k(B) & F_{k-1}(B) \end{pmatrix} \in M_{2m}(\F_2),
\end{equation}
where the Fibonacci polynomials naturally appear. The Fibonacci numbers satisfy several well-known divisibility
relations \cite[pp. 67--69]{Scheid} and we will need their counterparts for Fibonacci polynomials over
$K = \F_2$. We start with a definition.
\begin{definition}[Fibonacci index]\hfill\\
  The \emph{Fibonacci index} (sometimes \emph{depth}) of an irreducible polynomial
  $p \in K[x]$ is defined as the minimum number $d \in \N$, such that $p$ divides $F_d$.
\end{definition}
Several properties of these polynomials are discussed by Goldwasser et al. \mbox{\cite{GKT97,GKW02}};
for example, $F_n$ divides $F_m$, if and only if $n$ divides $m$ for $n,\,m \in \N$. In the proof of the
following lemma \cite[lem. 5/th. 7(a)]{GKW02}, we use the fact that there holds $F_{n-t}(x) + F_{n+t}(x)
= x F_n(x) F_t(x)$ over $\F_2$ for $n,\,t \in \N$ and $t \leq n$ \cite[lem. 4(2)]{GKT97}.
\begin{lemma}[Fibonacci polynomials and Fibonacci index]\label{lem:allpol}\hfill\\
  For every $m \in \N$, the product $F_{2^m+1}(x) \cdot F_{2^m-1}(x)$ is equal to the square of the
  product of all irreducible polynomials---except $x$---whose degree divides $m$.
  Therefore, the Fibonacci index of an irreducible polynomial $p \in \F_2[x]$ of degree~$m$ is well-defined
  and, for $p(x) \neq x$, is odd and divides either $2^m+1$ or $2^m-1$.
\end{lemma}
\begin{proof}
  We set $n = 2^m+1$ and $t = n-2$ and use that (i) over any finite field~$\F_q$ the product of all
  monic irreducible polynomials whose degree divide some $m \in \N$ is given by $x^{q^m}-x$
  \cite[th. 3.20]{LN08}, and (ii) $F_{2^m}(x) = x^{2^m-1}$ \cite[lem. A.3]{KRS10}. Note further that $F_{2^m-1}$
  and $F_{2^m+1}$ are coprime \cite[prop. 5]{GKT97}. \qed
\end{proof}
The following lemma was proven by Sutner \cite[th.~3.1]{S00}.
\begin{lemma}[Irreducible factors of $F_{2^m\pm 1}$]\label{thm:sutner}\hfill\\
  Let $p = \sum_{i = 0}^{m} a_i x^i \in \F_2[x]$ be an irreducible polynomial of order $m \in \N$. Then, $p$ divides
  $F_{2^m+1}$, if $a_1 = 1$; otherwise it divides $F_{2^m-1}$.
\end{lemma}

\subsection{Finding a stabilizer matrix}\label{subsec:findc}
Bandyopadhyay et al. \cite{BBRV02} showed that the problem of finding MUBs on $\C^d$ can be reduced to finding
certain partitions of operator bases of $M_d(\C)$ (cf. also the appendix). The problem of constructing cyclic MUBs
can accordingly be reduced to finding a suitable symplectic \emph{stabilizer matrix} $C \in M_{2m}(\F_2)$, where
$d = 2^m$ \cite{KRS10}. Numerically, it was shown that, at least for $m \leq 24$, this matrix may be chosen of the form
\begin{align}\label{eqn:defC}
  C = \begin{pmatrix} B & \Eins_m\\ \Eins_m & 0_m \end{pmatrix} \in M_{2m}(\F_2)
\end{align}
with $B \in M_m(\F_2)$, which we may call \emph{reduced stabilizer matrix}, and we will see in this article that this form
can be maintained for all $m \in \N$. (We summarize the explicit construction
of the generating unitary $U$ from such $B$ in the appendix.) Using the Fibonacci
polynomials, we demand that $C^{d+1}=\Eins_{2m}$, i.\,e.
\begin{align}
  C^{d+1} = \begin{pmatrix}
  F_{d+2}(B) & F_{d+1}(B)\\
  F_{d+1}(B) & F_{d}(B)
\end{pmatrix} = \begin{pmatrix}
  \Eins_m & 0_m\\
  0_m & \Eins_m
\end{pmatrix}.
\end{align}
We thus require $F_{d}(B)=\Eins_m$ and $F_{d+1}(B)=0_m$. By the recursion relation, the condition
$F_{d+2}(B) = \Eins_m$ is automatically fulfilled. The conditions on $B$ are the following \cite{KRS10}:
\begin{enumerate}[(i)]
  \item $B$ is symmetric, i.\,e. $B = B^t$,
  \item $F_k(B)$ is invertible for $k \in \MgE{d}$,
  \item $F_d(B) = B^{d-1} = \Eins_m$.
\end{enumerate}
Since the minimal polynomial of a matrix $B$ is the normalized polynomial of minimal degree that
annihilates $B$, we see that it is a factor of $F_{d+1}$.
\begin{lemma}[Characteristic polynomial of reduced stabilizer matrices]\label{CP-RSM}\hfill\\
  A matrix $B \in M_m(\F_2)$, $m \in \N$, that is at first annihilated by
  $F_{2^m - 1}$ or $F_{2^m + 1}$, has an irreducible characteristic polynomial. The characteristic
  polynomial thus coincides with the minimal polynomial of $B$.
\end{lemma}
\begin{proof}
  If the first Fibonacci polynomial that annihilates $B$ equals $F_{2^m \pm 1}$, $B$ is annihilated by
  a polynomial $p$ with Fibonacci index $2^m \pm 1$ that has degree $m$ and is therefore irreducible.
  Since the degree of $p$ coincides with the dimension of $B$, it is the irreducible characteristic polynomial
  of $B$. \qed
\end{proof}

\begin{lemma}[Multiplicative order of $B$]\label{lem:thirdCondition}\hfill\\
  If the characteristic polynomial of $B$ has Fibonacci index $d+1$, then $F_{d} (B) = \Eins_m$.
\end{lemma}
\begin{proof}
  By lemma \ref{CP-RSM}, $\chi_B$ is irreducible and its splitting field over $\F_2$ is isomorphic to $\F_d$
  with $d = 2^m$. By the Hamilton-Cayley theorem, $B$ is a root of $\chi_B$ and thus, as an element of the
  multiplicative group of a finite field of order $d$, there holds $B^{d-1} = \Eins_m$. The assertion now
  follows from $F_d(B)=B^{d-1}$ \cite[lem. A.3]{KRS10}. \qed
\end{proof}

\begin{proposition}[Representation of fields]\label{FieldRep}\hfill\\
  Let $L/K$ be a field extension of order $n \in \N$ and let $A \in M_n(K)$ be a matrix, such that
  its characteristic polynomial $\chi_A \in K[x]$ is irreducible (and therefore minimal). Then, the set
  $\Mge{f(A)}{f \in K[x]}$ is isomorphic to the extension field $L$.
\end{proposition}
Note that the set of polynomials in $A$ effectively includes only polynomials of degree less than $n$; it
also includes the zero polynomial with degree $-\infty$.
\begin{proof}
  By the Hamilton-Cayley theorem, there holds $\chi_A(A) = 0$ as matrix equality, i.\,e. $A$ is a root of
  the irreducible polynomial~$\chi$. By adjoining $A$ to the field $K$, we obtain the extension $L$
  \cite[p.\,66\,ff.]{LN08}. \qed
\end{proof}
With the help of this proposition and lemma \ref{lem:thirdCondition} we can replace conditions (ii) and (iii)
on $B$ from above by
\begin{enumerate}[(ii')]
 \item The Fibonacci index of the characteristic polynomial $\chi_B$ of $B$ is $d+1$.
\end{enumerate}
The following theorem shows that for every dimension $d = 2^m$ we can find a reduced stabilizer matrix $B$
and thus a unitary generator $U$ of a complete set of cyclic MUBs.
\begin{theorem}[Existence of reduced stabilizer matrices]\label{MainTheorem1}\hfill\\
 For any dimension $d = 2^m$ with $m \in \N$ there exists a reduced stabilizer matrix $B \in M_m(\F_2)$ that
 fulfills conditions (i) and (ii').
\end{theorem}
\begin{proof}
  We have to find a symmetric matrix $B \in M_m(\F_2)$ with irreducible characteristic polynomial and 
  Fibonacci index $d+1$. It is known that the number of polynomials in $\F_2[x]$ with degree $m$
  and Fibonacci index $d+1$ is given by $\frac{\phi(d+1)}{2m}$, where $\phi$ is Euler's totient function
  \cite[th. 8]{GKW02}. Since this expression is non-zero, there exists at least one polynomial for any dimension
  $d = 2^m$ that has maximal Fibonacci index $d+1$.
  \par It is well-known that any polynomial $p \in K[x]$ possesses a companion matrix, whose characteristic
  polynomial is $p$. For every monic polynomial over a finite field $K = \F_q$ there exists a symmetric matrix
  that is similar to the companion matrix and therefore has the same characteristic polynomial \cite[lem. 2]{BT98}.
  \qed
\end{proof}

\section{Explicit construction of the reduced stabilizer matrix for $m=2^k$}
In this section, we deal with a specific solution of $B$ for $m=2^k$ and $k \in \N$. We start by presenting
an iterated construction of $B$, which is followed by a discussion on the characteristic polynomial
and its Fibonacci index. We show that this construction is related to an open conjecture
in finite field theory by Wiedemann \cite{W88} (cf. also \cite[prob.~28]{MS95}), that is still of current
interest~\cite{V10}.
\par Let us recall some already known solutions \cite{KRS10}: For $k=1$ there exist two different forms of $B$
(by permutation), one of them being ${\tiny B_2 = \begin{pmatrix} 1 & 1 \\ 1 & 0 \end{pmatrix}} \in M_2(\F_2)$.
If we set $k=2$ a possible choice of $B$ is the block matrix ${\tiny B_4 = \begin{pmatrix} B_2 & \Eins\\
\Eins & 0 \end{pmatrix}} \in M_4(\F_2)$. By iterating this construction, we get a recursion
\begin{align}\label{eqn:B2k}
B_{2^k} = \begin{pmatrix} 
  B_{2^{k-1}} & \Eins\\
  \Eins & 0
\end{pmatrix} \in M_{2^k}(\F_2),
\end{align}
and we shall show that it is suitable for our purposes.

\subsection{Irreducible self-reciprocal polynomials and the characteristic polynomial of $B_{2^k}$}
In this subsection we want to evaluate the characteristic polynomials of the matrices $B_{2^k}$ of \eqref{eqn:B2k}.
It will turn out that the notion of reciprocal and self-reciprocal polynomials are of interest in the
context of determining $\chi_{B}$.
\begin{definition}[Reciprocal polynomial]\hfill\\
  The \emph{reciprocal} of a polynomial $f \in K[x]$ is defined by $f^*(x) := x^n f(x^{-1})$, where
  $n \in \N$ is the degree of $f$. If $f = f^*$, it is called \emph{self-reciprocal}.
\end{definition}
By applying the \emph{reciprocal operator} $Q$ on an arbitrary polynomial $f$ of degree $n$,
we get a self-reciprocal polynomial as
\begin{align}
  f^Q(x) := x^n f\left(x+x^{-1}\right).
\end{align}
Varshamov and Garakov \cite{VG69} have shown, that in the case of $f \in \F_2[x]$ and $f$ being irreducible,
$f^Q$ again is irreducible if and only if the linear coefficient of $f$ does not vanish; see also Meyn \cite{Meyn90}.
We can now relate the notion of reciprocal polynomials to our construction from \eqref{eqn:B2k}.
\begin{lemma}[Characteristic polynomials]\label{CharPol}\hfill\\
  Let $K$ be a finite field of characteristic 2 and $S \in M_n(K)$ with characteristic polynomial $\chi_S$.
  The characteristic polynomial of the matrix $S^\prime := {\tiny \begin{pmatrix} S & \Eins \\ \Eins & 0 \end{pmatrix}} \in M_{2n}(K)$
  is given by $\chi_{S^\prime} = (\chi_S)^Q$.
\end{lemma}
In the following proof, we use the well-known fact that for block matrices there holds
${\tiny \det \begin{pmatrix} A & B \\ C & D \end{pmatrix} = \det(D) \det(A - BD^{-1}C)}$, if $D$ is invertible,
which follows from the decomposition ${\tiny \begin{pmatrix} A & B \\ C & D \end{pmatrix} = \begin{pmatrix} \Eins & B
\\ 0 & D \end{pmatrix} \bigl(\begin{matrix} A - BD^{-1}C & 0 \\ D^{-1}C & \Eins \end{matrix}\bigr)}$, because the
determinant of block-triangular matrices is the product of the determinants of the blocks.
\begin{proof}
  We calculate $\chi_{S^\prime}(x) = \det(x\Eins_{2n} - S^\prime)
      = \det {\tiny \begin{pmatrix} x\Eins_n - S & \Eins_n \\ \Eins_n & x\Eins_n \end{pmatrix}}$
  which yields $\chi_{S^\prime}(x) = \det(x\Eins_n) \det(x\Eins_n - S - x^{-1}\Eins_n)
  = x^n \cdot \chi_S(x - x^{-1})$. Since $\mathrm{char}\,K = 2$, this is the reciprocal of $\chi_S$.\qed
\end{proof}
Given an arbitrary polynomial of degree $n$, $f(x) = \sum_{k = 0}^{n} a_k x^k$, after some easy calculations
we find $f^Q(x) = \sum_{k = 0}^{n}  \sum_{i = 0}^{k} \binom{k}{i} a_k x^{n+k-2i}$. A contribution to the linear coefficient
arises only from $k = i = n-1$, so that it is given by $a_{n-1}$.
A corollary thereof is the following.
\begin{corollary}[Minimal polynomial of $B_{2^k}$]\label{CharPolB}\hfill\\
 The characteristic polynomial of $B_{2^k} \in M_{2^k}(\F_2)$ as in \eqref{eqn:B2k}, $k \in \N$,
 is given by $f^{Q^k}(x)$ with $f(x)=1+x$. Therefore it is irreducible and minimal.
\end{corollary}
\begin{proof}
  By induction, the first part follows immediately from lemma \ref{CharPol}. According to Varshamov-Garakov,
  all polynomials are irreducible, provided that the linear coefficient does not vanish. This is guaranteed
  by the statements preceding this corollary.\qed
\end{proof}
This corollary implies that every $F_l(B_{2^k})$ is either invertible or zero, since by proposition \ref{FieldRep}
it is a representative of an element of some finite field. However, we do not yet know anything about the Fibonacci
index of the characteristic polynomial of $B_{2^k}$, but we shall relate it to an open conjecture in finite field theory
in the following subsection.

\subsection{Fibonacci index of $\chi_B$ and Wiedemann's conjecture}
Wiedemann \cite{W88} considered iterated quadratic extensions of $\F_2$ using generators
$x_{j+1} + x_{j+1}^{-1} = x_j$ for $j \in \N_0$ and $x_0 + x_0^{-1} = 1$; these extensions read $E_0 := \F_2(x_0)$,
$E_1 := E_0(x_1)$ etc., where $E_n$ is isomorphic to $\F_{2^{2^{n+1}}}$. He then showed that the order of $x_n$
divides the $n$-th Fermat number $\mathcal{F}_n = 2^{2^n} + 1$. Since the Fermat numbers are mutually coprime, the order
of $\tilde{x} := x_0x_1\dots x_n \in E_n$ is the product of the orders of the $x_i$. If the order of each $x_i$
is $\mathcal{F}_i$ for $i \leq n$, then $\tilde{x}$ is primitive in $E_n$, i.\,e. a generator of the multiplicative group of
$E_n$. This is \emph{Wiedemann's conjecture}, which is verified numerically for $n \leq 8$ \cite[p. 291/295]{W88}.
It turns out that our matrices $C_{2^k}$ can be seen as a realization of the $x_k$.
\begin{theorem}[Wiedemann analogy]\label{thm:wiedemann}\hfill\\
  Wiedemann's conjecture is true, if and only if the characteristic polynomial of every $B_{2^k}$ as
  in \eqref{eqn:B2k} has Fibonacci index $2^{2^k}+1$.
\end{theorem}
\begin{proof}
  Let $C_{2^k}$ be the stabilizer matrix as in \eqref{eqn:defC} using the construction of $B_{2^k}$ from
  \eqref{eqn:B2k}. Obviously, there holds $C_{2^0} + C_{2^0}^{-1} = \Eins$. Now, if $C_{2^k}$ is of order
  $2^{2^k}+1$, the characteristic polynomial of $B_{2^k}$ has Fibonacci index $2^{2^k}+1$. The characteristic
  polynomial of $C_{2^k}$ over the field generated by $B_{2^k} = C_{2^{k-1}}$ is then given by
  $x^2 + C_{2^{k-1}}x + 1$. Identifying $C_{2^k}$ with Wiedemann's $x_k$, the $C_{2^k}$ are the roots of
  these polynomials, which are equivalent to those in Wiedemann's construction. Thus, our problem is an instance
  of his conjecture. \qed
\end{proof}
We have thus shown that we are able to construct a reduced stabilizer matrix $B$, out of which we can construct a unitary operator
$U$ which generates a complete set of cyclic MUBs. In the appendix, we show that this operator can be written in the
form $U = \EZ^{\iE\psi} H^{\x m} P$ with some global phase $\EZ^{\iE\psi}$ and explicitly construct the diagonal
phase system $P$ from $B$ (see also \cite{KRS10}).

\section{Equivalence of sets of mutually unbiased bases}
In this section we shall discuss the equivalence of sets of mutually unbiased bases. At first, we will
introduce methods which are related to those of Calderbank et al. \cite{CCKS97}, Theorem 5.6 and Proposition~5.11;
see also Kantor \cite{Kan11}, Theorem 2.3. Our presentation is self-contained and stated in the language of
operators etc. rather than in terms of discrete geometry. We start with the usual definition of equivalence of MUBs.
\par We shall identify a basis of the Hilbert space $\cH = \C^d$ with a unitary matrix~$\cB$ consisting of the
basis vectors as columns and consider two sets of mutually unbiased bases, $\fS = \Mg{\cB_1,\,\dots,\cB_r}$ and
$\fS^\prime = \Mg{\cB_1^\prime,\,\dots,\cB_r^\prime}$; in our case, $r = d + 1$.
\begin{definition}[Equivalence of mutually unbiased bases]\hfill\\
  We say, the sets $\fS$ and $\fS^\prime$ are equivalent, if there exists a unitary operator~$U$ on~$\cH$,
  a permutation $\pi$ on $\MgE{r}$ and monomial matrices\footnote{A matrix $W$ is called monomial, if $W = D \Pi$,
  where $D = \mathrm{diag}(\lambda_1,\,\dots,\,\lambda_d)$ is diagonal and $\Pi$ a permutation; we assume
  $\betrag{\lambda_i} = 1$, i.\,e. $W$ is unitary.}
  $W_k$, $k \in \MgE{r}$, such that there holds $\cB_k^\prime = U \cB_{\pi (k)} W_k$ for all $k \in \MgE{r}$.
\end{definition}
The essential transformation is the unitary $U$; the matrices $W_k$ take care for the fact that for MUBs we
are more interested in a set of one-dimensional subspaces than in bases as ordered sets and $\pi$ represents
the fact that we consider unordered sets of bases.
\par As outlined in the appendix (for details cf. refs. \cite{BBRV02,KRS10}), our MUBs appear as eigenbases
of operators, i.\,e. $\fC = \Mg{\cC_1,\,\dots,\,\cC_{d+1}}$, where each class $\cC_k$ consists of
$d-1$ commuting Pauli operators. Therefore, $\cC_k = \Mg{\XZ(\vec{v}_1),\,\dots,\,\XZ(\vec{v}_{d+1})}$.
We can now reformulate equivalence in a way that the order of the basis vectors and their total phase---as in
the $W_k$---do not appear.
\begin{lemma}[Equivalence of mutually unbiased bases]\hfill\\
  The MUBs characterized by $\fC = \Mg{\cC_1,\,\dots,\,\cC_{d+1}}$ and
  $\fC^\prime = \Mg{\cC_1^\prime,\,\dots,\,\cC_{d+1}^\prime}$
  are equivalent, if and only if there exists a unitary $U$ on $\C^d$ and a permutation $\pi$ on $\MgE{d+1}$,
  such that $\cC_k^\prime = U\cC_{\pi(k)}U^\dagger$, where $U$ acts simultaneously on all operators within a class
  $\cC_k$.
\end{lemma}
\begin{proof}
Assume that the eigensystems of $\fC$ and $\fC^\prime$ are equivalent MUBs. As the ordering within a basis
does not appear here, the $W_k$ are irrelevant, and the unitary $U$ is the same as in the definition of
equivalence, since $Uv$ is an eigenvector of $UAU^\dagger$, if and only if $v$ is one of $A$. \qed
\end{proof}
As we assume $\fC$ and $\fC^\prime$ to consist of Pauli operators only, the conjugation by~$U$ must map Pauli
operators onto Pauli operators, i.\,e. $U$ belongs to the Clifford group. We can view the Pauli
operators---essentially the Pauli group factorized by its center
$\Mg{\pm 1,\,\pm\iE}\Eins$---as a projective representation of the symplectic space $V_S := \F_2^m \times \F_2^m$,
and the unitary $U$ thus corresponds to an invertible linear mapping $f$ on $V_S$, i.\,e. $f \in M_{2m}(\F_2)$.
In order to preserve the commutation relations, we require $f$ to be symplectic.
\par In general, the classes $\cC_k$ with $d-1 = 2^m-1$ elements are generated by $m$ elements. In terms of $V_S$
we have $m$ vectors which generate a subspace of $V_S$ with $d-1$ nonzero elements.
Viewing a class $\cC_k$ as a $2m \times (d-1)$-matrix, where the columns represent the Pauli operators
(as in the appendix), we find a
(non-unique) generator $G_k$, which is a $2m \times m$-matrix, and write $\Erz{G_k} = \cC_k$. As we can permute
the columns within the $\cC_k$ by a permutation $Q_k$, we find that $\fC$ and $\fC^\prime$ are equivalent, if and
only if there exists a mapping~$f$
and a permutation $\pi$ on $\MgN{d}$, such that $\Erz{G_k^\prime} = f \Erz{G_{\pi(k)}} Q_k$ for all $k$.\par
We can assume that the matrix representation of $f$ is of the block-matrix form ${\tiny\begin{pmatrix} s & t \\
u & v \end{pmatrix}}$; symplecticity is equivalent to the three conditions, that $s^tu$ and $t^tv$ are symmetric
and $v^ts -t^tu = \Eins_m$. Since in cyclic MUBs,
we can rotate through the bases, we can choose $\pi(0) = 0$, such that $\Erz{G_0} = \Erz{G_0^\prime}$ consists of
all $Z$-type Pauli operators with the standard basis as their eigenbasis. This implies $u=0$ and, to keep
symplecticity, that $v=(s^t)^{-1}$.
\par The elements within a class $\cC_k$ are given by \cite{KRS10}, i.\,e.
\begin{align}\label{eqn:classes}
  \cC_k = \{C^k \cdot (\Eins_m, 0_m)^t \cdot c : c \in \F_2^m \}.
\end{align}
The generators of the classes $\cC_k$ can be grouped into the generator of the $Z$-type Pauli operators
$G_0=(\Eins_m, 0_m)^t$ and the remaining generators $G_k=C^k \cdot (\Eins_m, 0_m)^t$ with $k \in \MgE{d}$,
which can be written as $G_k= (F_{k+1}(B), F_k(B))^t$ with the reduced stabilizer matrix $B$. For a complete set of MUBs, the generators
$G_k$ are constructed in a way, that for $k \in \MgE{d}$ the polynomial $F_k(B)$ is non-zero and invertible (cf. prop. \ref{FieldRep}).
Since any invertible operator that acts in \eqref{eqn:classes} on the elements $c$, is given by a permutation of those elements,
we will produce the same classes $\cC_k$, if we replace $G_k$ for $k \in \MgE{d}$ by
\begin{align}\label{StdForm}
 \bar G_k= \begin{pmatrix}p_k(B)\\ \Eins_m\end{pmatrix},
\end{align}
with $p_k(B)=F_{k+1}(B)\cdot F^{-1}_k(B)$. We call this the standard form, which fits into the construction of
\cite{BBRV02}. As $B$ is a root of its minimal polynomial with degree $m$, the $2^m$ different elements
$p_k(B)$ represent the extension field $\F_2^m$, i.\,e., the $p_k$ are all $2^m$ polynomials over $\F_2$ with
degree less than $m$.\par
By expressing the generators of a set of MUBs in the standard form we avoid the use of permutations $Q_k$.
Using the standard form, we can show to be free in the explicit choice of the symmetric matrix $B$
for a given irreducible polynomial.
\begin{lemma}[Equivalent sets of MUBs in certain cases]\label{EquivLemma}\hfill\\
  All symmetric matrices $B \in M_m(\F_2)$ with the same irreducible characteristic polynomial $p$
  generate equivalent sets of MUBs.
\end{lemma}
\begin{proof}
  Consider two symmetric matrices $B$ and $B^\prime$ with characteristic polynomial $p$. We can
  find some orthogonal matrix $s \in M_m(\F_2)$ (i.\,e. $s^t = s^{-1}$), such that $B^\prime = s B s^t$
  (orthogonal similarity). The matrix $f={\tiny\begin{pmatrix} s & 0 \\ 0 & s \end{pmatrix}}$ applied
  on the standard form $(p_k(B),\,\Eins)^t$ of eq. (\ref{StdForm}) results in $(s\,p_k(B),\,s)^t$ which
  in the standard form reads
  $(s\,p_k(B)\,s^t,\,\Eins)^t = (p_k(B^\prime),\,\Eins)^t$ showing equivalence in this case. \qed
\end{proof}
Using this lemma, we can relate even more sets, i.\,e. all sets which can be written in the standard form
of eq. \eqref{StdForm} with polynomials $p_k$ and some symmetric matrix $B$.

\pagebreak

\begin{theorem}[Equivalent sets of MUBs]\label{EquivTheorem}\hfill\\
 All complete sets of MUBs which are eigenbases of classes $\fC = \Mg{\cC_0,\,\dots,\,\cC_d}$
 generated by $G_0=(\Eins_m, 0_m)^t$ and $\bar G_k= (p_k(B), \Eins_m)^t$ for $k \in \MgE{d}$, some symmetric
 matrix $B \in M_m(\F_2)$ and $p_k$ being all $d = 2^m$ polynomials over $\F_2$ with degree
 less then $m$ are equivalent.
\end{theorem}
\begin{proof}
  Let $p$ and $p^\prime$ be two different irreducible polynomials of degree $m$ over the ground field $\F_2$
  with roots $\beta$ and $\beta^\prime$, respectively. By adjunction of $\beta$ to $\F_2$ we generate the field
  $\F_2(\beta) \cong \F_2[x]/p\F_2[x] \cong \F_{2^m}$. The same holds true for the adjunction of $\beta'$ in
  a similar fashion. Since both extensions are Galois extensions (cf. \cite[chap. 4.1]{Bosch}),
  the elements in $\F_2(\beta)$ and $\F_2(\beta')$ are equivalent up to
  permutation. If we represent $\beta$ as a symmetric matrix $B \in M_m(\F_2)$, such that the minimal polynomial of $B$
  equals the minimal polynomial of $\beta$, we are free in the explicit choice of $B$ as seen in lemma \ref{EquivLemma}.
\qed
\end{proof}
As mentioned in \cite[chap. 4.1]{BBRV02}, the choice of a basis to produce the elements $p_k(B)$ implies a field
structure and generates the MUBs of Wootters and Fields \cite{WF89}. By theorem \ref{EquivTheorem} they are
equivalent to the construction given here, and by \cite{GR09} they are also equivalent to the bases which were
generated by Klappenecker and R{\"o}tteler \cite{KR04} and by Bandyopadhyay et al. \cite{BBRV02}. 
Thus, many of the known MUBs can be constructed with cyclic structure, which yields to a short description
and allows a simple implementation in experiments like the quantum circuit given in \cite{SR11}.

\section{Conclusions}
In this paper we have shown how to relate the problem of constructing unitary generators of cyclic mutually unbiased bases to Fibonacci
polynomials and their properties. In particular, we have shown that the problem for $d = 2^m$ can be reduced to
finding a symmetric matrix with irreducible characteristic polynomial of degree $m$ with non-zero linear
coefficient, which is known to exist. If $m$ itself is of the form $m = 2^k$, we give an explicit construction of
such matrix, provided Wiedemann's conjecture is true.
We have proven that all complete sets of cyclic mutually unbiased bases our scheme provides, are equivalent and belong
to the same class as the mutually unbiased bases by Wootters and Fields.
The methods of this work are used in another article \cite{SR11}, where we give a simple direct construction
of the unitary generator $U$, which can be applied without reference to finite field theory and show how this can be translated into
a quantum circuit, which may be useful in experiments.

\appendix
\section{Construction of the MUB-generating unitary operator}\label{app:unitary}
In this appendix we will review and slightly reformulate the construction of the unitary matrix $U \in M_d(\C)$
from $C \in M_{2m}(\F_2)$ or $B \in M_m(\F_2)$, respectively, up to a global phase $\EZ^{\iE\psi}$ \cite{KRS10}.
\par On the Hilbert space $\cH = \C^p$, $p \in \N$ prime, with canonical orthonormal basis
$\Mge{\ket{i}}{i \in \F_p}$, we define
\begin{equation}
  X\ket{i} := \ket{i+1} \quad\text{and}\quad Z\ket{i} := \omega_p^i\ket{i}
\end{equation}
with $\omega_p := \exp(2\pi\iE/p)$.
For $\vec{a} := (\vec{a}_z|\vec{a}_x)^t \in \F_p^{2m}$ with \mbox{$\vec{a}_\ny = (a_1^{(\ny)},\,\dots,\,a_m^{(\ny)})^t \in \F_p^m$} and $\ny \in \Mg{z,\,x}$ we define the set of \emph{Pauli operators} (which is not the Pauli group!)
by $P_p^m := \Mge{\XZ(\vec{a})}{\vec{a} \in \F_p^{2m}}$ with
\begin{equation}
  \XZ(\vec{a}) := f(\vec{a}) \cdot X^{a_1^{(x)}}Z^{a_1^{(z)}} \x \dots \x X^{a_m^{(x)}}Z^{a_m^{(z)}},
\end{equation}
where $f(\vec{a}) := \iE^{a_1^{(x)}a_1^{(z)} + \dots + a_m^{(x)}a_m^{(z)}}$, if $p = 2$, and $f(\vec{a}) = 1$
otherwise. This implies $f(\vec{a})^{-1}f(\vec{b})^{-1}\XZ(\vec{a})\XZ(\vec{b})
= \omega_p^{a_1^{(z)}b_1^{(x)} + \dots + a_m^{(z)}b_m^{(x)}} f(\vec{a}+\vec{b})^{-1}\XZ(\vec{a}+ \vec{b})$,
and it follows $\XZ(\vec{a})\XZ(\vec{b}) = \omega_p^{\symp{\vec{a}}{\vec{b}}}
\XZ(\vec{b})\XZ(\vec{a})$, where $\Fkt{\symp{\,\cdot\,}{\,\cdot\,}}{\F_p^{2m} \times \F_p^{2m}}{\F_p}$
with $\symp{\vec{a}}{\vec{b}} := \sum_{k = 1}^{n} a_k^{(z)} b_k^{(x)} - a_k^{(x)} b_k^{(z)}$ denotes the
\emph{symplectic product}.
\par Consider the vectors $\vec{z}_i,\,\vec{x}_i \in \F_p^{2m}$, $i \in \MgN{m-1}$, where the $i$-th component
of the first resp. second half of the vector is one and all other entries vanish. These vectors form
a \emph{symplectic} (or \emph{hyperbolic}) \emph{basis} of $\F_p^{2m}$, i.\,e. for all $i,\,j \in \MgE{p-1}$
there hold the equations $\symp{\vec{x}_i}{\vec{x}_j} = \symp{\vec{z}_i}{\vec{z}_j} = 0$ and
$\symp{\vec{x}_i}{\vec{z}_j} = \delta_{ij}$.
For a matrix $C \in M_{2m}(\F_p)$ the images $\vec{\ny}_i^\prime := C \vec{\ny}_i^\prime$, $\ny \in \Mg{z,\,x}$,
form a symplectic basis, if and only if $C$ is a symplectic matrix, i.\,e. $C^tSC = S$ for
$S = {\tiny\begin{pmatrix} 0_m & \Eins_m \\ -\Eins_m & 0_m \end{pmatrix} \in M_{2m}(\F_p)}$.
\par For any symplectic basis, we can define \emph{logical operators} by $\overline{X}_i := \XZ(C\vec{x}_i)$
and $\overline{Z}_i := \XZ(C\vec{z}_i)$ for $C$ as above. By definition, the \emph{logical state}
$\ket{0}_L$ is the joint eigenstate of all $\overline{Z}_i$ with eigenvalue $+1$, and we get the other logical
states by $\ket{j}_L = \overline{X}_1^{j_1} \x \dots \x \overline{X}_1^{j_m}\ket{0}_L$ for
$j = (j_1,\,\dots,\,j_m)^t \in \F_2^m$. If we choose the block matrix form $C = {\tiny\begin{pmatrix} B & \Eins_m \\ \Eins_m & 0_m \end{pmatrix}} \in M_{2m}(\F_2)$ with $B \in M_m(\F_2)$, the matrix $C$ is symplectic, if and only
if $B$ is symmetric, and we simply have $\overline{Z}_i = X_i$ for all $i \in \MgE{m}$ and $\ket{0}_L
= 2^{-m/2}\sum_{j \in \F_2^m}\ket{j}$ (unique up to a global phase).
\par Given a partition $P_p^m = \Mg{\Eins} \cup \cC_0 \cup \dots \cup \cC_d$, where the sets $\cC_k$ are
pairwise disjoint and contain $d-1$ commuting operators each, the eigenbases of the $\cC_k$ are
mutually unbiased \cite{BBRV02}. Let $\cC_0 := \Mge{\XZ(\vec{a}_z|0)}{\vec{a}_z \in \F_2^m \setminus \Mg{0}}$
be the non-identity \emph{$Z$-type operators} and suppose $U \in M_d(\C)$ is unitary such that
\mbox{$U\cC_kU^\dagger = \cC_{(k+1) \mod (d+1)}$} for all $k \in \MgN{d}$; then, $U^{d+1} = U^0 = \Eins_d$,
and the columns of $U^k$, $k \in \MgN{d}$, are the eigenbases of the $\cC_k$, i.\,e., $U$ is the generator
of cyclic MUBs. (In particular, $U$ belongs to the \emph{Clifford group}, i.\,e. the group of unitary operators
which leave the set of Pauli operators invariant, possibly up to a $p$-th root of unity: $\mathfrak{C}_p^m
:= \bigl\{ U \in \mathfrak{U}(p^m) \,|\, (\forall \vec{a} \in \F_p^{2m}) \bigr. $
$\bigl. (\exists \vec{b} \in \F_p^{2m},\,k \in \F_p) (U \XZ(\vec{a}) U^\dagger = \omega_p^k \XZ(\vec{b})) \bigr\}$.)
Thus, for all $j \in \F_2^m$ there holds $U\ket{j} = \ket{j}_L = \prod_{k = 1}^{m} \XZ(C\vec{x}_k)^{j_k} \ket{0}_L$,
and
\begin{equation}\label{ConstrU}
  U = 2^{-m/2}\sum\nolimits_{i,j \in \F_2^m} \prod\nolimits_{k = 1}^{m} \XZ(C\vec{x}_k)^{j_k} \ket{i}\bra{j}.
\end{equation}
With $B = (b_{ij})_{i,j = 1}^{m}$ as above, we have $\XZ(C\vec{x}_k) = (\bigotimes_{i = 1}^{k-1}
X^{b_{ik}}) \x \iE^{b_{kk}} X^{b_{kk}} Z \x (\bigotimes_{i = k+1}^{m} X^{b_{ik}})$, and the $k$-th tensor
factor in \eqref{ConstrU} is given by
$X^{b_{1k}j_1 + \cdots + b_{kk}j_k} \cdot Z^{j_k} \cdot X^{b_{k+1,k}j_{k+1} + \cdots + b_{mk}j_m}$, where
we can shift the factor $Z$ to the beginning, which yields
$(-1)^{b_{1k}j_1j_k + \dots + b_{kk}j_k^2} \cdot Z^{j_k} \cdot X^{b_{1k}j_1 + \cdots + b_{mk}j_m}$.
By defining new abbreviations
$p_j := \iE^{\sum_{k = 1}^{m} b_{kk}j_k} (-1)^{\sum_{k = 1}^{m} b_{1k}j_1j_k + \dots + b_{kk}j_k^2}$
and $\tilde{X}_j := \bigotimes_{k = 1}^{m} X^{b_{1k}j_1 + \cdots + b_{mk}j_m}$, we find
\begin{equation}
  U = 2^{-m/2}\sum\nolimits_{i,j \in \F_2^m} p_j (Z^{j_1} \x \dots \x Z^{j_m}) \tilde{X}_j \ket{i}\bra{j};
\end{equation}
we can pull through the sum over $i$ and see that $\sum_i \ket{i}$ is invariant under $\tilde{X}_j$, which thus
can be absorbed. This results in $U = 2^{-m/2} \sum\nolimits_{ij} p_j (-1)^{ij} \ket{i}\bra{j}$
or, using the Hadamard matrix $H = {\tiny\begin{pmatrix} 1 & 1 \\ 1 & -1 \end{pmatrix}}/\sqrt{2} \in M_2(\C)$,
in $U = H^{\x m}P$, with diagonal \emph{phase system} matrix $P = \mathrm{diag}\bigl((p_j)_{j \in \F_2^m}\bigr)$.
Since $B$ is symmetric, we find $(-1)^{\sum_{k = 1}^{m} b_{1k}j_1j_k + \dots + b_{kk}j_k^2}
= \iE^{\bra{j}B\ket{j}} \iE^{\sum_{k = 1}^{m} b_{kk}j_k^2}$ with $\bra{j}B\ket{j}
= {\sum_{k,\,l = 1}^{m} b_{kl}j_k j_l}$ being a quadratic form, which leads to
$p_j := \iE^{\bra{j}B\ket{j}} (-1)^{\sum_{k = 1}^{m} b_{kk}j_k^2}$.

\begin{acknowledgement}
  The authors acknowledge funding by CASED and by BMBF/QuOReP.
\end{acknowledgement}

\end{document}